\documentclass[11pt,letterpaper,english]{article}
\usepackage[utf8]{inputenc}			
\usepackage{makeidx}
\usepackage[pdftex]{graphicx}
\usepackage{mathtools,amsfonts}
\usepackage{array}
\usepackage{multirow, hhline}

\usepackage{nicefrac}
\usepackage{tikz}
\usetikzlibrary{arrows}
\usetikzlibrary{calc}
\usepackage{paralist}

\usepackage{wrapfig}
\usepackage[english]{babel}		
\usepackage[margin=1in]{geometry}	
\usepackage{lipsum}			
\usepackage{comment}
\usepackage{chngpage}			
\usepackage{amsfonts,amssymb,amsmath, amsthm}
\usepackage{amsopn}
\usepackage{color}
\usepackage{pgfplots}
\usepackage{pgfgantt}
\usepackage{algpseudocode}
\usepackage[Algorithm]{algorithm}
\usepackage{subcaption}
\usepackage{tabto}
\usepackage{csquotes} 
\usepackage{multirow}
\usepackage{amsthm}
\usepackage[colorlinks,urlcolor=black,citecolor=black,linkcolor=black,menucolor=black]{hyperref}

\newcommand{\abs}[1]{| #1 |}

\newtheorem{lemma}{Lemma}
\newtheorem{theorem}{Theorem}
\newtheorem{observation}[theorem]{Observation}
\newtheorem{proposition}[theorem]{Proposition}
\newtheorem{corollary}{Corollary}
\newtheorem{definition}{Definition}
\newtheorem{claim}{Claim}
\newtheorem{example}{Example}
\usepackage{xfrac}

\newcommand{\HA}[1]{\textcolor{blue}{Hana: #1}}

\newcommand{\Hana}{MMS-feasible }

\newcommand{\prob}{\mathbf{P}}

\usepackage{authblk}
\usepackage{hyperref}

\title{EFX Allocations: Simplifications and Improvements}
\author[1]{Hannaneh Akrami}
\affil[1]{Max Planck Institute for Informatics and Graduiertenschule Informatik, Universit\"at des Saarlandes} 
\author[2]{Noga Alon}
\affil[2]{Princeton University}
\author[3]{Bhaskar Ray Chaudhury}
\affil[2]{University of Illinois, Urbana-Champaign}
\author[3]{Jugal Garg}
\author[4]{Kurt Mehlhorn}
\affil[4]{Max Planck Institute for Informatics and Universit\"at des Saarlandes}
\author[3]{Ruta Mehta}
\date{}

\begin{document}
\maketitle

\begin{abstract}
    
The existence of EFX allocations is a fundamental open problem in discrete fair division. Given a set of agents and  indivisible goods, the goal is to determine the existence of an allocation where {\em no agent envies another following the removal of any single good} from the other agent's bundle. Since the general problem has been illusive, progress is made on two fronts: $(i)$ proving existence when the number of agents is small, $(ii)$ proving existence of relaxations of EFX. In this paper, we improve results on both fronts (and simplify in one of the cases). 

\cite{ChaudhuryGM20} showed the existence of EFX allocations when there are three agents with additive valuation functions. The proof in~\cite{ChaudhuryGM20} is long, requires careful and complex case analysis, and does not extend even when one of the agents has a general monotone valuation function. We prove the existence of EFX allocations with three agents, \emph{restricting only one agent to have an additive valuation function (the other agents may have any monotone valuation functions)}. Our proof technique is significantly simpler and shorter than the proof in~\cite{ChaudhuryGM20} and therefore more accessible. {In particular, it does not use the concepts of \emph{champions, champion-graphs, half-bundles} (in contrast to the algorithms in~\cite{CKMS21,ChaudhuryGM20, CGMMM21}) and \emph{envy-graph} (in contrast to most algorithms that prove existence of relaxations of envy-freeness, including weaker relaxations like EF1).} Our technique also extends to settings when two agents have any monotone valuation function and one agent has  an {\em \Hana}valuation function (a strict generalization of \emph{nice-cancelable} valuation functions~\cite{BergerCFF} which subsumes additive, \emph{budget-additive} and \emph{unit demand} valuation functions). This takes us a step closer to resolving the existence of EFX allocations when all three agents have general monotone valuations.

Secondly, we consider relaxations of EFX allocations, namely, approximate-EFX allocations and EFX allocations with few unallocated goods (charity).~\cite{CGMMM21} showed the existence of $(1-\varepsilon)$-EFX allocation with $\mathcal{O}((n/\varepsilon)^{\sfrac{4}{5}})$ charity by establishing a connection to a problem in extremal combinatorics. We improve the result in~\cite{CGMMM21} and  prove the existence of $(1-\varepsilon)$-EFX allocations with $\mathcal{\tilde{O}}((n/ \varepsilon)^{\sfrac{1}{2}})$ charity. In fact, some of our techniques can be used to prove improved upper-bounds on a problem in \emph{zero-sum combinatorics} introduced by Alon and Krivelevich~\cite{AK21, MS21}.

\end{abstract}
\section{Introduction}
Fair division has been a fundamental branch of mathematical economics over the last seven decades (since the seminal work of Hugo Steinhaus in the 1940s~\cite{Steinhaus48}). In a classical fair division problem, the goal is to ``fairly" allocate a set of items among a set of agents. Such problems find very early mentions in history, for instance,  in ancient Greek mythology and the Bible. Even more so today, many real-life scenarios are paradigmatic of the problems in this domain, e.g., division of family inheritance~\cite{PrattZ90}, divorce settlements~\cite{BramsT96}, spectrum allocation~\cite{EtkinPT05}, air traffic management~\cite{Vossen02}, course allocation~\cite{BudishC10} and many more\footnote{Check~\cite{spliddit} and~\cite{fairoutcome} for more detailed explanation of fair division protocols used in day to day problems.}. For the past two decades, the computer science community has developed concrete formulations and tractable solutions to fair division problems and thus contributing substantially to the development in the field. With the advent of the Internet and the rise of  centralized electronic platforms that intend to impose fairness constraints on their decisions (e.g., Airbnb would like to fairly match hosts and guests, and Uber would like to fairly match drivers and riders etc..), there has been an increasing demand for computationally tractable protocols to solve fair division problems.

In this paper, we focus on one of the important open problems in discrete fair division. In a classical setting of discrete fair division, we have a set $[n]$ of $n$ agents and a set $M$ of $m$ indivisible goods. Each agent $i$ is equipped with a valuation function $v_i \colon 2^M \rightarrow \mathbb{R}_{\geq 0}$ which captures the utility agent $i$ derives from any bundle that can be allocated to {her}.  One of the most well studied classes of valuations are \emph{additive valuations}, i.e.,  $v_i(S) = \sum_{g \in S} v_i(\{g\})$ for all $S \subseteq M$. The goal is to determine a partition $X = \langle X_1, X_2, \dots ,X_n \rangle$ of $M$ such that $X_i$ is allocated to agent $i$ which is \emph{fair}. Depending on the notion of fairness used, there are several different problems in this setting.

\paragraph{Envy-freeness up to any good (EFX)} The quintessential notion of fairness is that of envy-freeness. An allocation $X = \langle X_1, X_2, \dots ,X_n \rangle$ is envy-free if every agent prefers her bundle as much as she prefers the bundle of any other agent, i.e., $v_i(X_i) \geq v_{i}(X_{i'})$ for all $i,i' \in [n]$. However, an envy-free allocation does not always exist, e.g., consider dividing a single valuable good among two agents. In any feasible allocation, the agent with no good will envy the agent that has been allocated one good. This necessitates the study of relaxed notions of envy-freeness. In this paper, we consider the relaxation known as \emph{envy-freeness up to any good} (EFX). An allocation $X = \langle X_1, X_2, \dots, X_n \rangle$ is EFX if and only if for all pairs of agents $i$ and $i'$, we have $v_i(X_i) \geq v_i(X_{i'} \setminus \{g\})$ for all $g \in X_{i'}$, i.e., the envy should disappear following the removal of any single good from $i$'s bundle. EFX is in fact considered to be the ``closest analogue of envy-freeness'' in discrete fair division~\cite{CaragiannisGravin19}.  Unfortunately, the existence of EFX allocations is still unsettled despite significant effort by several researchers~\cite{Moulin19, CaragiannisKMP016}  and is considered one of the most important open problems in fair division~\cite{ProcacciaCACM20}. There have been studies on 
\begin{itemize}
    \item the existence of EFX allocations in restricted settings. In particular, EFX existence has been studied when there are small number of agents~\cite{TimPlaut18, ChaudhuryGM20}, and when agents have specific valuation functions~\cite{HalperPPS20}. 
    \item The existence of relaxations of EFX allocations has also been investigated, e.g., approximate EFX allocations~\cite{TimPlaut18, AmanatidisMN20},   EFX with bounded charity~\cite{CKMS21, BergerCFF}, approximate EFX with bounded charity~\cite{CGMMM21}.  
\end{itemize}
Improving the understanding in both the above settings is a systematic direction towards the big problem. We first mention the existing results in the above two settings and mention some of their pitfalls. Thereafter, we highlight main results of this paper and show how they address the said pitfalls. In particular, we focus on the existence of EFX allocations with small number of agents and the existence of approximate EFX allocations with bounded charity.


\paragraph{Existence of EFX Allocations with Small Number of Agents.} Plaut and Roughgarden~\cite{TimPlaut18} first showed the existence of EFX allocations when there are two agents using the \emph{cut and choose protocol}. The existence of EFX allocations gets notoriously more difficult with three or more agents. The existence of EFX allocations with three agents was shown by Chaudhury et al~\cite{ChaudhuryGM20}. The proof of existence in ~\cite{ChaudhuryGM20} involves several new concepts like \emph{champions}, \emph{champion-graphs} and \emph{half-bundles}, spans over 15 pages, and requires a lot of careful and detailed case analysis. Furthermore, the proof technique does not extend to the setting with four or more agents~\cite{CGMMM21}. We articulate the primary bottleneck here: At a high-level, the algorithm in~\cite{ChaudhuryGM20} moves in the space of partial EFX allocations\footnote{EFX allocations where not all goods are allocated.} iteratively improving the vector $\langle v_1(X_1), v_2(X_2), v_3(X_3) \rangle$ lexicographically, where $v_i( \cdot )$ is the valuation function of agent $i$. However,~\cite{CGMMM21} exhibit an instance with four agents, nine goods and a partial EFX allocation $X$ such that in any complete EFX allocation $X'$, $v_1(X'_1) < v_1(X_1)$, i.e., agent 1 (which is the highest priority agent) is better off in $X$ than in any complete EFX allocation. {All of this} necessitates the study of a different approach for the existence of EFX allocations. As the first main contribution of this paper,  we present a new proof for the existence of EFX allocations for three agents, which is significantly shorter and simpler (we do not use the notions of champions, champion-graphs and half-bundles) than the proof in~\cite{ChaudhuryGM20}. Our approach is algorithmic, but in contrast to the approach in~\cite{ChaudhuryGM20},  our algorithm moves in the space of complete allocations (instead of partial allocations) iteratively improving a certain potential as long as the current allocation is not EFX.  Furthermore, the algorithm also allows us to prove the existence of EFX  \emph{beyond additivity}, i.e.,  even when only one of the agents has an additive valuation function and the other agents have general monotone valuation functions, our algorithm can determine an EFX allocation. We note that the proof in~\cite{ChaudhuryGM20} crucially needs all the valuation functions to be additive.

\begin{theorem}
\label{EFXmainthmintro}
EFX allocations exist with three agents as long as there is at least one agent with an additive valuation function.
\end{theorem}

Berger et al.~\cite{BergerCFF} show the existence of EFX allocations for three agents when agents have more general valuation functions, called \emph{nice-cancelable valuation functions} (defined formally in Section~\ref{prelim}). Nice-cancelable valuation functions generalize many well known valuation functions like \emph{additive, budget-additive, unit-demand} and more. We introduce a class of valuation functions called \emph{\Hana valuation functions} (defined formally in Section~\ref{prelim}) that are very natural in the fair division setting and they \emph{strictly} generalize nice-cancelable valuations. Our proof of existence also holds when two agents have general valuation functions and one of the agents has an \Hana valuation function. Thus, we also prove,

\begin{theorem}
\label{EFXmainthmintro2}
EFX allocations exist with three agents as long as there is at least one agent with an \Hana valuation function.
\end{theorem}

\paragraph{Existence of Approximate EFX with Bounded Charity.} Caragiannis et al.~\cite{CaragiannisGravin19} introduced the notion of EFX with charity. The goal here is to find ``good'' partial EFX allocations, i.e., partial EFX allocations where the set of unallocated goods are not very valuable. In particular, they show that there always exists a partial EFX allocation $X$ such that for each agent $i$, we have $v_i(X_i) \geq 1/2 \cdot v_i(X^*_i)$, where $X^* = \langle X^*_1, X^*_2, \dots, X^*_n \rangle$ is the allocation with maximum \emph{Nash welfare}\footnote{The Nash welfare of any allocation $Y$ is the geometric mean of the valuations of the agents, $\big(\prod_{i \in [n]} v_i(Y_i) \big)^{\sfrac{1}{n}}$. It is often considered a direct measure of the fairness and efficiency of an allocation.}. Following the same line of work, Chaudhury et al.~\cite{CKMS21} showed the existence of a partial EFX allocation $X$ such that no agent {envies} the set of unallocated goods and the total number of unallocated goods is at most $n-1 \ll m$. Quite recently, Chaudhury et al.~\cite{CGMMM21} showed the existence of a $(1-\varepsilon)$-EFX allocation with $\mathcal{O}((n/\varepsilon)^{\sfrac{4}{5}})$ charity, where an allocation $X$ is said to be $(1-\varepsilon)$-EFX if and only if $v_i(X_i) \geq (1-\varepsilon) \cdot v_{i}(X_{i'} \setminus \{g\})$ for all $g \in X_{i'}$. While the last result is not a strict improvement of the result in~\cite{CKMS21} (since it ensures $(1-\varepsilon)$-EFX instead of exact EFX), it is the best relaxation of EFX that we can compute in polynomial time, as the algorithm in ~\cite{CKMS21} can only be modified to give $(1-\varepsilon)$-EFX with $n-1$ charity in polynomial time. Another key aspect of the technique in~\cite{CGMMM21} is the reduction of the problem of improving the bounds on charity to a purely graph theoretic problem. In particular~\cite{CGMMM21} define the notion of a \emph{rainbow cycle number}: Given an integer $d > 0$, the rainbow cycle number $R(d)$ is the largest $k$ such that there exists a $k$-partite graph $G =(V_1 \cup V_2 \cup \dots \cup V_k, E)$ such that 
\begin{itemize}
    \item each part has at most $d$ vertices, i.e., $\lvert V_i \rvert \leq d$, and 
    \item every vertex in $G$ has exactly one incoming edge from every part in $G$ except the part containing it, and 
    \item  there exists no cycle $C$ in $G$ that visits each part at most once.
\end{itemize}
Let $h^{-1}(d)$ denote the smallest integer $\ell$ such that $h(\ell) = \ell \cdot R(\ell) \geq d$. Then there always exist an $(1-\varepsilon)$-EFX allocation with $\mathcal{O}(\frac{n}{\varepsilon \cdot h^{-1}(n / \varepsilon)})$. So smaller the upper bound on $h(\ell)$, lower is the number of unallocated goods. ~\cite{CGMMM21} show that $R(d) \in \mathcal{O}(d^4)$ and thus establish the existence of $(1-\varepsilon)$-EFX allocation with $\mathcal{O}((n/\varepsilon)^{\sfrac{4}{5}})$ charity.  An upper  bound of $ \mathcal{O}(d^2 2^{(\log\log d)^2})$ was obtained by~\cite{BBK22}, thereby showing the existence of EFX allocations with $\mathcal{O}((n / \varepsilon)^{0.67})$ charity. In this paper, we close this line of improvements by proving an almost tight upper bound on $d$ (matching the lower bound up to a log factor). 

\begin{theorem}\label{rainbow-bound}
    Given any integer $d>0$, the rainbow cycle number $R(d) \in \mathcal{O}(d \log d)$.
\end{theorem}

As a consequence of the improved bound we obtain:

\begin{theorem}\label{eps-efx}
There exists a polynomial time algorithm that determines a partial $(1-\varepsilon)$-EFX allocation $X$ such that no agent envies the set of unallocated goods and the total number of unallocated goods is $\mathcal{\tilde{O}}((n/\varepsilon)^{1/2})$. Furthermore, $\mathit{NW}(X) \geq 1/2e^{\sfrac{1}{e}} \cdot \mathit{NW}(X^*)$ where $X^*$ is the allocation with maximum Nash welfare. 
\end{theorem}

\paragraph{Rainbow Cycle and Zero-sum Combinatorics.} We believe that investigating tighter bounds on $R(d)$ is interesting in its own right. Quite recently, Berendsohn, Boyadzhiyska, and Kozma~\cite{BBK22} showed intriguing connections between rainbow cycle number and zero sum problems in  extremal combinatorics. Zero sum problems in graphs ask questions of the following flavor: Given an edge/vertex weighted graph, whether there exists a certain substructure (for example cliques, cycles, paths etc.) with a zero sum (modulo some integer). In particular,~\cite{BBK22} show that the rainbow cycle number is a natural generalization of the zero sum problems studied in Alon and Krivelevich~\cite{AK21}, and M{\'{e}}sz{\'{a}}ros and Steiner~\cite{MS21}. Both papers~\cite{AK21, MS21} aim to upper bound the maximum number of vertices of a complete bidirected graph with integer edge labels avoiding a zero sum cycle (modulo $d$).~\cite{BBK22} show through a simple argument that this is upper bounded by the \emph{permutation rainbow cycle number} $R_p(d)$, which is defined by introducing an additional constraint in the definition of $R(d)$ that for all $i,j$, each vertex in $V_i$ has exactly one \emph{outgoing} edge to some vertex in $V_j$ (in addition to exactly one incoming edge from some vertex in $V_j$). In Section~\ref{perm-rainbowcyclenumer}, we show through a simple argument that $R_p(d) \leq 2d-2$, thereby also improving the upper bounds of $\mathcal{O}(d \log (d))$ in~\cite{AK21} and $8d-1$ in~\cite{MS21}.

\begin{lemma}
\label{permutationlemma}
We have $R_p(d) \leq 2d-2$. Therefore, by the Observation made by~\cite{BBK22}, the maximum number of vertices of a complete bidirected graph with integer edge labels avoiding a zero sum cycle (modulo $d$) is at most $2d-2$.
\end{lemma}

\subsection{Further Related Work}
Fair division has received significant attention since the seminal work of Steinhaus~\cite{Steinhaus48} in the 1940s. Other than envy-freeness, another fundamental fairness notion is that of \emph{proportionality}. Recall that, in an envy-free allocation, every agent values her own bundle at least as much as she values the bundle of any other agent. However, in a proportional allocation, each agent gets a bundle that she values $1/n$ times her valuation on the entire set of goods. Since envy-freeness and proportionality cannot always be guaranteed while dividing indivisible goods, various relaxations of the same have been studied. Alongside EFX, another popular relaxation of envy-freeness is \emph{envy-freeness up to one good (EF1)} where no agent envies another agent following the removal of \emph{some} good from the other agent's bundle. While the existence of EFX allocations is open, EF1 allocations are known to exist for any number of agents, even when agents have general monotone valuation functions~\cite{LiptonMMS04}. While EF1 and EFX are fairness notions that relax envy-freeness, the most popular notions of fairness that relaxes proportionality for indivisible goods are \emph{maximin share} (MMS), proportionality up to one good (PROP1), proportionality up to any good (PROPx), and proportionality up to the maximin good (PROPm). The MMS was introduced by Budish~\cite{budish2011combinatorial}. While MMS allocations do not always exist~\cite{KPW18}, there has been extensive work to come up with approximate MMS allocations~\cite{budish2011combinatorial,BL16,AMNS17,BK17,KPW18,GhodsiHSSY18,JGargMT19,GargT19}. On the other hand, PROPx is stronger than PROPm, which is stronger than PROP1. While PROPx allocations do not always exist~\cite{Moulin19}, PROPm allocations are guaranteed to exist~\cite{BaklanovGGS21}. Some works assume ordinal ranking over the goods, as opposed to cardinal values, e.g.,~\cite{AzizGMW15,BramsKK17}. 


Alongside fairness, the efficiency of an allocation is also a desirable property. Two common measures of efficiency is that of Pareto-optimality and Nash welfare. Caragiannis et al.\ \cite{CaragiannisKMP016} showed that any allocation that has the maximum Nash welfare is guaranteed to be Pareto-optimal (efficient) and EF1 (fair).  Barman et al.~\cite{BKV18} give a pseudo-polynomial algorithm to find an allocation that is both EF1 and Pareto-optimal. Other works explore relaxations of EFX with high Nash welfare~\cite{CaragiannisGravin19, CKMS21}. 

\paragraph{Independent Work.} Independently and concurrently to our work,~\cite{BBK22} also investigate upper bounds on rainbow cycle number. They obtain the same upper bound  of $2d-2$ for $R_p(d)$. 


\section{Preliminaries}
\label{prelim}
An instance of discrete fair division is given by the tuple $\langle [n], M, \mathcal {V} \rangle$, where $[n]$ is the set of agents, $M$ is the set of indivisible goods and $\mathcal V = (v_1(\cdot), v_2(\cdot), \dots, v_n(\cdot))$ where each $v_i: 2^M \rightarrow \mathbb{R}_{\geq 0}$ denotes the valuation of agent $i$. Typically, the valuation functions are assumed to be \emph{monotone}, i.e., for each agent $i$, $v_i(S \cup \{g\}) \geq v_i(S)$ for all $S \subseteq M$ and $g \notin S$. A valuation $v_i(\cdot)$ is said to be \emph{additive} if $v_i(S) = \sum_{g \in S} v_i(\{g\})$ for all $S \subseteq M$. For ease of notation, we use $g$ instead of $\{g\}$. We also use $S \oplus_i T$ for $v_i(S) \oplus v_i(T)$ with $\oplus \in \left\{\leq, \geq ,< , > \right\}$.

Given an allocation $X = \langle X_1, X_2, \dots, X_n \rangle$, we say that an agent $i$ \emph{strongly envies} an agent $i'$ if and only if $v_i(X_i) < v_i(X_{i'} \setminus \{g\})$ for some $g \in X_{i'}$. Thus, an allocation is an EFX allocation if there is no strong envy between any pair of agents. We now introduce certain definitions and recall certain concepts that will be useful in the upcoming sections. 

\begin{definition}[EFX feasibility]
Given a partition $X = (X_1, X_2, \dots, X_n)$ of $M$, a bundle $X_k$ is EFX-feasible  to agent $i$ if and only if $X_k \geq_i \mathit{max}_{j \in [n]} \mathit{max}_{g \in X_j} X_j \setminus g$. Therefore an allocation $X = \langle X_1, X_2, \dots, X_n \rangle$ is EFX if for each agent $i$, $X_i$ is EFX-feasible . 
\end{definition}

Chaudhury et al.~\cite{ChaudhuryGM20} introduced the notion of non-degenerate instances where no agent values two distinct bundles the same. They showed that to prove the existence of EFX allocations in the additive setting, it suffices to show the existence of EFX allocations for all non-degenerate instances. We adapt their approach and show that the same claim holds, even when agents have general monotone valuations.

\paragraph{Non-Degenerate Instances~\cite{ChaudhuryGM20}} We call an instance $I = \langle [n], M, \mathcal{V} \rangle$ non-degenerate if and only if no agent values two different sets equally, i.e., $\forall i \in [n]$ we have $v_i(S) \neq v_i(T)$ for all $S \neq T$. We extend the technique in~\cite{ChaudhuryGM20} and  show that it suffices to deal with non-degenerate instances when there are $n$ agents with general valuation functions, i.e., if there exists an EFX allocation in all non-degenerate instances, then there exists an EFX allocation in all instances. We defer the reader to the appendix for the detailed proof.

\textit{Henceforth, we assume that the given instance is non-degenerate, implying that all goods are positively valued by all agents.}

\paragraph{\Hana valuations.} In this paper, we introduce a new class of valuation functions called \Hana valuations which are natural extensions of additive valuations in a fair division setting. 

\begin{definition} \label{Hana-func}
    A valuation function $v: 2^M \rightarrow \mathbb{R}_{\geq 0}$ is \Hana if for every subset of goods $S \subseteq M$ and every partitions $A = (A_1, A_2)$ and $B = (B_1, B_2)$ of $S$, we have
    $$\max(v(B_1), v(B_2)) \geq \min(v(A_1), v(A_2)).$$
\end{definition}

Informally, these are the valuations under which, an agent always has a bundle in  any 2-partition that she values more than her MMS value, i.e., given an agent $i$ with an \Hana valuation $v(\cdot)$, in any 2-partition of  $S \subseteq M$, say $B = (B_1, B_2)$, we have $\mathit{max}(v(B_1), v(B_2)) \geq \mathit{MMS}_i (2,S)$, where $\mathit{MMS}_i(2,S)$ is the MMS value of agent $i$ on the set $S$ when there are $2$ agents. Also, note that if there are two agents and one of the agents has an \Hana valuation function, then irrespective of the valuation function of the other agent, MMS allocations always exist: Consider an instance where agent 1 has an \Hana valuation function  and agent 2 has a general monotone valuation function. Consider agent 2's  MMS optimal partition of the good set  $A= (A_1, A_2)$. Let agent 1 pick her favorite bundle from $A$. Then, agent 1 has a bundle that she values at least as much as her MMS value (as she has an \Hana valuation function), and agent 2 has a bundle that she values at least as much as her MMS value as $A$ is an MMS optimal partition according to agent 2. 

\Hana valuations {strictly} generalize the \emph{nice-cancelable valuation functions} introduced in~\cite{BergerCFF}. A valuation function $v: 2^M \rightarrow \mathbb{R}_{\geq 0}$ is nice-cancelable if for every $S, T \subset M$ and $g \in M \setminus (S \cup T)$, we have $v(S \cup \{g\}) > v(T \cup \{g\}) \Rightarrow v(S) > v(T)$. Nice-cancelable valuations include \emph{budget-additive} ($v(S) = \mathit{min} (\sum_{s \in S} v(s) , c)$), \emph{unit demand} ($v(S) = \mathit{max}_{j \in S} v(s)$), and \emph{multiplicative} ($v(S) = \prod_{s \in S} v(s)$) valuations~\cite{BergerCFF}. 

\begin{lemma}\label{general-nice}
    Every nice-cancelable function is \Hana.
\end{lemma}
\begin{proof}
    We first make an observation about a nice-cancelable valuation function.
    
    \begin{observation}\label{nice-cancel-obs}
    If $v$ is a nice-cancelable valuation, then for every $S, T \subset M$ and $Z \subseteq M \setminus (S \cup T)$, we have $v(S \cup Z) > v(T \cup Z) \Rightarrow v(S) > v(T)$. 
\end{observation}

    Let $v$ be a nice-cancelable function. For a subset of goods $S \subseteq M$, consider any two partitions $A = (A_1, A_2)$ and $B = (B_1, B_2)$ of $S$. Without loss of generality assume $v(A_1 \cap B_1) < v(A_2 \cap B_2)$. Since $(A_1 \cap B_2)$ is disjoint from  $(A_1 \cap B_1) \cup (A_2 \cap B_2)$, by the contrapositive of Observation \ref{nice-cancel-obs} applied to nice-cancelable valuation $v$, we have, 
    \begin{align}
        v((A_1 \cap B_1) \cup (A_1 \cap B_2)) &< v((A_2 \cap B_2) \cup (A_1 \cap B_2)). \label{eq1}
    \end{align}
    Therefore,
    \begin{align*}
        \min(v(A_1), v(A_2)) &\leq  v(A_1) \\
        &= v((A_1 \cap B_1) \cup (A_1 \cap B_2)) &\hbox{$A_1 = (A_1 \cap B_1) \cup (A_1 \cap B_2)$}\\
        &< v((A_2 \cap B_2) \cup (A_1 \cap B_2)) &\hbox{Inequality \eqref{eq1}}\\
        &= v(B_2) &\hbox{$B_2 = (A_2 \cap B_2) \cup (A_1 \cap B_2)$}\\
        &\leq \max(v(B_1), v(B_2)).
    \end{align*}
\end{proof}

{In order to prove that \Hana functions strictly generalize nice-cancelable functions, we present an example of a valuation function which is \Hana but not nice-cancelable.
\begin{example}\label{example}
    Let $M = \{g_1, g_2, g_3\}$. The value of $v(S)$ is given in Table \ref{table} for all $S \subseteq M$. 
    \begin{table}
		\centering
		\begin{tabular}{ c|c c c c c c c c}
			$S$ & $\{g_1\}$ & $\{g_2\}$ & $\{g_3\}$ & $\{g_1, g_2\}$ & $\{g_1, g_3\}$ & $\{g_2, g_3\}$ & $\{g_1, g_2, g_3\}$ &\\
			\hline
			$v$ & $1$ & $2$ & $3$ & $10$ & $4$ & $5$ & $13$\\
		\end{tabular}
		\caption{valuation function $v$ is \Hana but not nice-cancelable.}
		\label{table}
	\end{table}
	First note that $v(g_1 \cup g_2) > v(g_3 \cup g_2)$ but $v(g_1) < v(g_3)$. Therefore, $v$ is not nice-cancelable. Now we prove that $v$ is \Hana. Let $S \subseteq M$ and $A = (A_1, A_2)$, $B = (B_1, B_2)$ be two partitions of $M$. Without loss of generality, assume $|A_1| \leq |A_2|$. If $A_1 = \emptyset$, $\min (v(A_1, v(A_2))) = 0 \leq \max(v(B_1), v(B_2))$. Hence, we assume $|A_1| \geq 1$ and therefore, we have $|S| \geq 2$. Moreover, if $A = B$, then $\max(v(B_1), v(B_2)) = \max(v(A_1), v(A_2)) \geq \min(v(A_1), v(A_2))$. Thus, we also assume $A \neq B$. If $S = \{g, g'\}$, the only two different possible partitioning of $S$ is $A = (\{g\}, \{g'\})$ and $B = (\emptyset, \{g, g'\})$. For all $g, g' \in M$, $v(\{g, g'\}) > \max(v(g), v(g'))$. Therefore, $\max(v(B_1), v(B_2)) \geq \min(v(A_1), v(A_2))$. 
	If $S = \{g_1, g_2, g_3\}$, then $|A_1| = 1$ and therefore, $\min(v(A_1), v(A_2)) \leq v(A_1) \leq \max_{g \in M} (v(g)) = 3$. Without loss of generality, let $g_3 \in B_1$. For all $T \subseteq M$ such that $g_3 \in T$, we have $v(T) \geq 3$. Thus, $\max(v(B_1), v(B_2)) \geq v(B_1) \geq 3 \geq \min(v(A_1), v(A_2))$. 
\end{example}
Lemma \ref{strict-general} follows from Lemma \ref{general-nice} and Example \ref{example}.
\begin{lemma}\label{strict-general}
    The class of \Hana valuation functions is a strict superclass of nice-cancelable valuation functions.
\end{lemma}
}

\paragraph{Preliminaries on Rainbow Cycle Number.}~\cite{CGMMM21} reduce the problem of finding approximate EFX allocations with sublinear charity to a problem in extremal graph theory. In particular, they introduce the notion of a rainbow cycle number.

\begin{definition}\label{rainbow-def}
    Given an integer $d > 0$, the rainbow cycle number $R(d)$ is the largest $k$ such that there exists a $k$-partite graph $G =(V_1 \cup V_2 \cup \dots \cup V_k, E)$ such that 
    \begin{itemize}
        \item each part has at most $d$ vertices, i.e., $\lvert V_i \rvert \leq d$, and 
        \item every vertex has exactly one incoming edge from every part other than the one containing it\footnote{In the original definition of the rainbow cycle number $R(d)$ in~\cite{CGMMM21}, every vertex can have more than one incoming edge from a part. However, by reducing the number of edges, we can only increase the upper-bound on $R(d)$.}, and 
        \item  there exists no cycle $C$ in $G$ that visits each part at most once.
    \end{itemize}
    {We also refer to cycles that visit each part at most once as ``rainbow'' cycles.}
\end{definition}

They show that any finite upper bound on $R(d)$ implies the existence of approximate EFX allocations with sublinear charity. Better upper bounds on $R(d)$ would give us better bounds on the charity. In particular, they prove the following theorem.

\begin{theorem}{~\cite{CGMMM21}}
\label{theorem-reduction}
Let $G = (V_1 \cup V_2 \cup \dots V_k, E)$ be a $k$-partite digraph such that (i) each part has at most $d$ vertices and (ii) each vertex in $G$ has an incoming edge from every part other than the one containing it. Furthermore, let $k > T(d) \geq R(d)$. If there exists a polynomial time algorithm that can find a cycle visiting each part at most once in $G$ , then there exists a polynomial time algorithm that determines a partial EFX allocation $X$ such that 
\begin{itemize}
    \item the total number of unallocated goods is in $\mathcal{O}(n/\varepsilon h^{-1}(n / \varepsilon))$ where $h^{-1}(d)$ is the smallest integer $\ell$ such that $h(\ell) = \ell \cdot T(\ell) \geq d$.
    \item $\mathit{NW}(X) \geq 1 / (2e^{\sfrac{1}{e}}) \cdot \mathit{NW}(X^*)$, where $X^*$ is the allocation with maximum Nash welfare.
\end{itemize}
\end{theorem}
\section{Technical Overview}
\label{techoverview}
In this section, we briefly highlight the main technical ideas used to show our results.

\subsection{EFX existence beyond additivity.}
We present an algorithmic proof for the existence of EFX allocations when agents have valuations more general than additive valuations. The main takeaway of our algorithm is that it does not require the sophisticated concepts introduced by the techniques in~\cite{CKMS21, ChaudhuryGM20} that rely on improving a potential function while  moving in the space of partial EFX allocations.  In fact, our algorithm does not even require the concept of an envy-graph which is a very fundamental concept used by the algorithms in~\cite{CKMS21, ChaudhuryGM20} and also by~\cite{TimPlaut18, LiptonMMS04} to prove the existence of weaker relaxations of envy-freeness (like EF1 and  $1/2$-EFX).

The crucial idea in our technique is to move in the space of partitions (of the good set), improving a certain potential as long as we cannot find an EFX allocation from the current partition, i.e., we cannot find a \emph{complete} allocation of the bundles in the partition such that the EFX property is satisfied. In particular, we always maintain a partition  $X = ( X_1, X_2, X_3)$ such that (i) agent 1 finds $X_1$ and $X_2$ EFX-feasible  and (ii) at least one of agent 2 and agent 3 finds $X_3$ EFX-feasible. Note that such allocations always exist: Agent 1 can determine a partition such that all bundles are EFX-feasible  for her (such a partition is possible as agent 1 can run the algorithm in~\cite{TimPlaut18} to find an EFX allocation assuming all three agents have agent 1's valuation function, thereby making all bundles EFX-feasible  for her) and we call agent 2's favorite bundle in the partition $X_3$ (which is obviously EFX-feasible  for her) and the remaining bundles $X_1$ and $X_2$ arbitrarily. Then, we have a partition that satisfies the invariant.

Note that if any one agent 2 or 3 finds one of $X_1$ or $X_2$ EFX-feasible, then we easily get an EFX allocation. Indeed, assume w.l.o.g.\  that agent 2 finds $X_3$ EFX-feasible. 
Now, if 
\begin{itemize}
    \item agent 3 finds $X_2$ EFX-feasible , then we have an EFX allocation: agent 1 $\gets X_1$, agent 2 $ \gets X_3$, and agent 3 $\gets X_2$. We can give a symmetric argument when agent 3 finds $X_1$ EFX-feasible.
    \item Similarly, if agent 2 finds $X_2$ EFX-feasible, then we can let agent 3 pick her favourite bundle in the partition (which is obviously EFX-feasible  for her) and still give agents 1 and 2 an EFX-feasible  bundle. We can give a symmetric argument when agent 2 finds $X_1$ EFX-feasible.
\end{itemize}
Therefore, we only need to consider the scenario where only $X_3$ is EFX-feasible for agents 2 and 3. Essentially, in this scenario, $X_3$ is ``too valuable'' to agents 2 and 3, as they do not find any of the remaining bundles EFX-feasible. \emph{A natural attempt would be to remove some good(s) from $X_3$ and allocate it to $X_1$ or $X_2$, i.e., we can increase the valuation of agent 1 for her EFX-feasible  bundle(s) by removing the excess goods allocated to the only EFX-feasible  bundle of agents 2 and 3.} This brings us to our potential function: $ \phi(X)  = \mathit{min}(v_1(X_1), v_1(X_2))$. Now, the non-triviality lies in determining the set of goods to be removed from $X_3$, and then allocating them to $X_1$ and $X_2$ such that we maintain our invariants. Although non-trivial, this turns out to be significantly simpler than the procedure used in~\cite{ChaudhuryGM20} and also holds when agents 1 and 2 have general monotone valuation functions and agent 3 has an \Hana valuation function. The entire procedure is elaborated in Section~\ref{EFXsimple}.

\subsection{Improved Bounds on Rainbow Cycle Number.}
Our technique to achieve the improved bound involves the probabilistic method. It is significantly simpler and yields better guarantees. We briefly sketch our algorithmic proof. Let there be $k$ parts in $G = (V_1 \cup V_2 \cup \dots V_k, E)$. Note that each part has at most $d$ vertices and each vertex has at least one incoming edge from every part. We pick one vertex $v_i$ from each part $V_i$ uniformly and independently at random. Now, it suffices to show that with non-zero probability, the induced graph on the vertices $v_1, v_2, \dots, v_k$ is cyclic for some $k \in \mathcal{O}(d \log d)$.  Note that if every vertex in $G[v_1, \dots, v_k]$ has an incoming edge, then $G[v_1 \dots v_k]$ is cyclic. So we need to show a non-zero lower bound on the probability of the each vertex having at least one incoming edge or equivalently show an upper bound on the probability that each vertex has no incoming edge n $G[v_1\dots v_k]$. To this end, let $E_{v_i}$ denote the event that vertex $v_i$ has no incoming edge in $G[v_1 \dots v_k]$. Note that $\prob[E_{v_i}] \leq (1- 1/d)^{k-1}$: $v_i$ has at least one incoming edge from each part and therefore the probability that there is no incoming edge from $v_j$ to $v_i$ is at most $(1-1/d)$ for all $j$. Since all $v_j$'s are independently chosen, the probability that $v_i$ has no incoming edge from any part is at most $(1-1/d)^{(k-1)}$.  Then, by union bound, $\prob[\cup_{i \in [n]} E_{v_i}] \leq \sum_{i \in [n]} \prob[E_{v_i}] \leq k(1-1/d)^{(k-1)}$. Therefore, the probability that  $G[v_1 \dots v_k]$ is cyclic is at least $1 - k(1-1/d)^{(k-1)}$ which is strictly positive for $k \in \mathcal O(d \log d)$.

\section{EFX Existence beyond Additivity}
\label{EFXsimple}
Before we give the new algorithm, we first give the reader a quick recap of the Plaut and Roughgarden algorithm~\cite{TimPlaut18} (PR algorithm) that determines an EFX allocation when all agents have the same valuation function, $v( \cdot )$ (the only assumption on $v(\cdot)$ is that it is monotone). The algorithm starts with any arbitrary allocation $X$ (which may not be EFX), and makes minor reallocations to improve the valuation of the agent who has the lowest value, i.e., it modifies $X$ to $X'$ such that $\mathit{min}_{i \in [n]} v(X'_i) > \mathit{min}_{i \in [n]} v(X_i)$. We now elaborate on the reallocation procedure: Let $\ell$ be the agent with the lowest valuation in $X$. If $X$ is not EFX, then there exists agents $i$ and $j$ such that $v(X_i) < v(X_j \setminus \{g\})$ for some $g \in X_j$. Since $v(X_{\ell}) < v(X_i)$, we also have $v(X_{\ell}) < v(X_j \setminus \{g\})$. The algorithm removes the good $g$ from $j$'s bundle and allocates it to $\ell$. Observe that $v(X_k) > v(X_{\ell})$ for all $k \neq \ell$ as we assume non-degeneracy. Also, we have $v(X_{\ell} \cup \{g\})$ and $v(X_j \setminus \{g\})$ greater than $v(X_{\ell})$. Therefore, the valuation of every new bundle is strictly larger than the valuation of $X_{\ell}$. Therefore, the valuation of the agent with the lowest valuation improves. This implies that the reallocation procedure will never revisit a particular allocation and as a result this process will eventually converge to an EFX allocation $Y$ with $v(Y_i) > v(X_{\ell})$ for all $i \in [n]$. Formally, 

\begin{lemma}[\cite{TimPlaut18}]
\label{PlautRoghgarden}
Let $X = ( X_1, X_2, X_3 )$ be an arbitrary 3-partition. Running the PR algorithm with any monotone valuation $v$ results in an EFX-partition $X' = ( X'_1, X'_2, X'_3 )$ with ${\min(v(X_1), v(X_2), v(X_3)) \le \min(v(X'_1), v(X'_2), v(X'_3))}$. We have equality only if the input is already EFX with respect to $v$.
\end{lemma}

In contrast to the algorithms in~\cite{ChaudhuryGM20,CKMS21, BergerCFF, TimPlaut18}, our algorithm moves in the space of complete EFX allocations iteratively maintaining some invariants. As long as our allocation is not EFX, we make some reallocations to the existing allocation and improve a certain potential. We give the proof here assuming only monotonicity for the valuation functions of agents 1 and 2 and assuming MMS-feasibility for the valuation of agent 3, i.e., $v_1(\cdot)$ and $v_2(\cdot)$ are general monotone valuation functions and $v_3(\cdot)$ is MMS-feasible.  We now elaborate our algorithm. We maintain a partition $(X_1, X_2, X_3)$ of the good set such that  

\begin{itemize}
\label{invariants}
\item $X_1$ and $X_2$ are EFX-feasible for agent 1.
\item $X_3$ is EFX-feasible for at least one of agents $2$ and $3$.
\end{itemize}

One can show the existence of allocations satisfying the above invariants by running the PR algorithm and initializing:  Agent 1 runs the PR algorithm with $v=v_1$ to determine a partition $(X_1, X_2,X_3 )$ such that all the three bundles are EFX-feasible  for her. Then, agent 2 picks her favorite bundle out of the three, say $X_3$. Clearly, $X_3$ is EFX-feasible  for agent 2, and $X_1$ and $X_2$ are EFX-feasible  for agent 1. Thus $X$ satisfies the invariants. 

We define our potential function as $\phi(X) = \min(v_1(X_1), v_1(X_2))$. We now elaborate how to modify $X$ and improve the potential when we cannot determine an EFX allocation from the partition $X$, i.e., we cannot determine an allocation of the bundles in $X$ to the agents that satisfies the EFX property.



\subsection{Reallocation when we cannot get an EFX allocation from $X$}

Let $X =  (X_1, X_2, X_3)$ be a partition satisfying the invariants. Without loss of generality, let us assume that agent 2 finds $X_3$ EFX-feasible. Observe that if any one of agents 2 or 3 finds bundles $X_1$ or $X_2$ EFX-feasible, then we are done: If agent 3 finds one of $X_1$ or $X_2$ EFX-feasible, then we can allocate agent 3's EFX-feasible  bundle to her, $X_3$ to agent 2 and the remaining bundle of $X_1$ and $X_2$ to agent 1 and get an EFX allocation. Similarly, if agent 2 finds $X_1$ or $X_2$ EFX-feasible, we ask agent 3 to pick her favourite bundle out of $X_1$, $X_2$ and $X_3$. Now, note that no matter which bundle agent 3 picks, there is always a way to allocate agents 1 and 2 their EFX-feasible bundles as agent 1 finds $X_1$ and $X_2$ EFX-feasible and agent 2 finds $X_3$ and at least one of $X_1$ or $X_2$ EFX-feasible\footnote{If agent 3 picks $X_1$, allocate $X_2$ to agent 1 and $X_3$ to agent 2. If agent 3 picks $X_2$, then allocate $X_1$ to agent 1 and $X_3$ to agent 2. Finally, if she picks $X_3$, then allocate the bundle among $X_1$ and $X_2$ that is EFX-feasible for agent 2 to agent 2 and the remaining bundle to agent 1.}. Therefore, from here on we assume that neither agent 2 nor agent 3 finds $X_1$ or $X_2$ EFX-feasible. Let $g_i$ be the good in $X_3$ such that $X_3 \setminus g_i \geq_i X_3 \setminus h$ for all $h \in X_3$, i.e., $X_3 \setminus g_i$ is the most valued proper subset of $X_3$ for agent $i$.

\begin{observation}
\label{noEFXfeasibility}
 For $i \in \{2,3\}$, we have $X_3 \setminus g_i >_i \mathit{max}_i(X_1,X_2)$.
\end{observation}

\begin{proof}
We prove for $i=2$. The proof for $i=3$ is identical. Let us assume otherwise and say w.l.o.g. $X_1 >_2 X_3 \setminus g_2$. Then, the only reason why $X_1$ is not EFX-feasible  for agent 2 is if $X_1 <_2 X_2 \setminus g$ for some $g \in X_2$. But, in that case, we have $X_2 >_2 X_1 >_2 X_3 \setminus g_2$. Therefore, we have $X_2 >_2 \mathit{max}_{\ell \in [3]} \mathit{max}_{h \in X_{\ell}} X_{\ell} \setminus h$, implying that $X_2$ is EFX-feasible, which is a contradiction. 
\end{proof}

W.l.o.g. assume that $X_1 <_1 X_2$, implying that $\phi(X) = v_1(X_1)$. We now distinguish two cases depending on how valuable the bundle $X_1 \cup g_i$ is to agent $i$ for $i \in \{2,3\}$ and give the appropriate reallocations in both cases. In particular, we first look into the case where $X_3 \setminus g_i$  is still more valuable to agent $i$ than $X_1 \cup g_i$ for at at least one $i \in \{2,3\}$.

\paragraph{Case: $X_3 \setminus g_2 >_2 X_1 \cup g_2$ or $X_3 \setminus g_3 >_3 X_1 \cup g_3$, i.e., $X_3 \setminus g_i$ is the favorite bundle for agent $i$ in the partition $X_1 \cup g_i$, $X_2$ and $X_3\setminus g_i$ for at least one $i \in \{2,3 \}$.} We provide the reallocation rules assuming that $X_3 \setminus g_2 >_2 X_1 \cup g_2$. The rules for the case $X_3 \setminus g_3 >_3 X_1 \cup g_3$ is symmetric. Now, consider the partition $ (X_1 \cup g_2, X_2, X_3  \setminus g_2)$.

By Observation~\ref{noEFXfeasibility}, $X_3 \setminus g_2 >_2 X_2$ and by our current case $X_3 \setminus g_2 >_2 X_1 \cup g_2$, implying that $X_3 \setminus g_2$ is an EFX-feasible  bundle for agent 2. Let $X'_1$ be a minimal subset of $X_1 \cup g_2$ w.r.t.\ set inclusion that agent 1 values more than $X_1$, i.e., agent 1 values $X_1$ more than any proper subset of $X'_1$ and $X'_1 >_1 X_1$. Let $X'_2 = X_2$ and $X'_3 = (X_3 \setminus g_2) \cup ((X_1 \cup g_2) \setminus X'_1)$. We define the partition $X' = (X'_1, X'_2, X'_3)$. Observe that $\phi(X') > \phi(X)$ as $X'_2 = X_2 >_1 X_1$ (by assumption) and $X'_1 >_1 X_1$ (by definition). Also note that $X'_3$ is EFX-feasible  for agent 2 as it is the most valuable bundle in $X'$ for agent 2. Now, if $X'_1$ and $X'_2$ are EFX-feasible  for agent 1, then all the invariants are maintained and we are done. So now we look into the case when at least one of $X'_1$ and $X'_2$ is not EFX-feasible  for agent 1 in $X'$.

We first make an observation on agent 1's valuation on the bundles $X'_1$ and $X'_2$.

\begin{observation}
\label{X1X2EFXfeasible}
We have $X'_1 >_1 X'_2 \setminus g$ for all $g \in X'_2$ and $X'_2 >_1 X'_1 \setminus h$ for all $h \in X'_1$.
\end{observation}

\begin{proof}
Note that $X'_1 >_1 X_1$ by definition of $X'_1$ and $X_1 >_1 X_2 \setminus g$ for all $g \in X_2$ as $X_1$ was EFX-feasible  for agent 1 in $X$. Since $X'_2 = X_2$, we have $X'_1 >_1 X'_2 \setminus g$ for all $g \in X'_2$.

Similarly, $X_2 >_1 X_1$ by assumption. Furthermore $X_1 >_1 X'_1 \setminus h$ for all $h \in X'_1$ by the definition of $X'_1$. Since $X'_2 = X_2$, we have $X'_2 >_1 X'_1 \setminus h$ for all $h \in X'_1$.
\end{proof}

By Observation~\ref{X1X2EFXfeasible}, if $X'_1$ and $X'_2$ are not EFX-feasible  for agent 1 in $X'$, then $X'_3 \setminus g >_1 \mathit{min}_1(X'_1,X'_2)$ for some $g \in X'_3$. However, in that case, we run the PR algorithm on the partition $X'$ with agent 1's valuation. Let $Y = (Y_1, Y_2, Y_3)$ be the final partition at the end of the PR algorithm. We have,  
\begin{align*}
    \mathit{min}(v_1(Y_1), v_1(Y_2), v_1(Y_3)) &> \mathit{min}(v_1(X'_1), v_1(X'_2), v_1(X'_3)) &\text{(by Lemma~\ref{PlautRoghgarden})}\\
                                               &= \mathit{min}(v_1(X'_1), v_1(X'_2)) &\text{(as $v_1(X'_3) > \mathit{min}(v_1(X'_1),v_1(X'_2))$)}\\ 
                                               &= \phi(X') \\
                                               &>\phi(X)
\end{align*}
We then let agent 2 pick her favorite bundle out of $Y_1, Y_2$ and $Y_3$. Let us assume w.l.o.g. that she chooses $Y_3$. Then, allocation $Y$ satisfies the invariants and we have $\phi(Y) = \mathit{min}(v_1(Y_1), v_1(Y_2)) \geq \mathit{min}(v_1(Y_1), v_1(Y_2), v_1(Y_3)) > \phi(X)$. Thus, we are done.

\paragraph{Remark:} Note that we have not used the MMS-feasibility of $v_3(\cdot)$ yet. All the arguments in this case hold when all three valuation functions are general monotone. We use MMS-feasibility crucially in the upcoming case.

\paragraph{Case: $X_3 \setminus g_2 <_2 X_1 \cup g_2$ and $X_3 \setminus g_3 <_3 X_1 \cup g_3$, i.e., $X_1 \cup g_i$ is the favourite bundle in the partition $X_1 \cup g_i$, $X_2$ and $X_3 \setminus g_i$ for all $i \in \{2,3\}$:} From Observation~\ref{noEFXfeasibility}, we have $X_3 \setminus g_i >_i X_2$ for $i \in \{2,3\}$. Therefore, we have,

\[ X_2 <_2 X_3 \setminus g_2 <_2 X_1 \cup g_2 \quad\text{and}\quad X_2 <_3 X_3 \setminus g_3 <_3 X_1 \cup g_3.\]
By MMS-feasibility of valuation function $v_3(\cdot)$, we conclude that $X_2 <_3 \mathit{max}_3(Z,Z')$ where $(Z, Z')$ is any valid 2-partition of the good set $X_1 \cup X_3$, as MMS-feasibility implies that $\mathit{max}_3(Z, Z') \geq \mathit{min}_3(X_1 \cup g_3, X_3\setminus g_3)>_3 X_2$.  We run the PR algorithm on the 2-partition $(X_1 \cup g_2, X_3 \setminus g_2)$ with agent 2's valuation ($v_2(\cdot)$)\footnote{Note that this time we run the PR algorithm with $n=2$ as opposed to the usual $n=3$ in the prior cases.}. Let $(Y_2,Y_3)$ be the output of the PR algorithm.  We let agent 3 choose her favorite among $Y_2$ and $Y_3$. Assume w.l.o.g. she chooses $Y_3$. Now, consider the allocation $X'$

\[  \text{agent 1}:\ X_2    \quad   \text{agent 2}:\ Y_2    \quad \text{agent 3}: \ Y_3.\]

We now analyze the strong envy in the allocation. To this end, we first observe that agents 2 and 3 do not strongly envy anyone.

\begin{observation}
\label{Y2Y3feasible}
 $Y_2$ is EFX-feasible  for agent 2 and $Y_3$ is EFX-feasible  for agent 3 in $X'$.
\end{observation}

\begin{proof}
 Since $(Y_2,Y_3)$ is the output of the PR algorithm run on $(X_1 \cup g_2, X_3 \setminus g_2)$ with agent 2's valuation function, (i) $Y_2 >_2 Y_3 \setminus h$ for all $h \in Y_3$, and (ii) $Y_2 \geq \mathit{min}_2(X_1 \cup g_2, X_3 \setminus g_2) >_2 X_2$, where the first inequality follows from Lemma~\ref{PlautRoghgarden} and the second inequality follows from the fact that $X_1 \cup g_2 >_2 X_3 \setminus g_2 >_2 X_2$. Therefore $Y_2$ is EFX-feasible  w.r.t. agent 2.

 Now, we look into agent 3. Note that $Y_3 = \mathit{max}_3(Y_2, Y_3)$ as agent 3 picks her favourite among $Y_2$ and $Y_3$. Therefore $Y_3 >_3 Y_2$\footnote{Strict inequality follows due to non-degeneracy.}. Furthermore, due to the MMS-feasibility of $v_3(\cdot)$ and the fact that $(Y_2, Y_3)$ is a valid 2 partition of the good set $X_1\cup X_3$, we have  $Y_3 = \mathit{max}_3(Y_2, Y_3) >_3 X_2$. Therefore, $Y_3>_3 \mathit{max}_3(Y_2,X_2)$ and thus is an EFX-feasible  bundle for agent 3.
\end{proof}

Therefore, the only possible strong envy is from agent 1. We now enlist the possible strong envy that may arise from agent 1 and also show corresponding reallocations.

\begin{itemize}
    \item Agent 1 does not strongly envy $Y_2$ and $Y_3$: Then we are done as $X'$ is an EFX allocation.
    
    \item Agent 1 strongly envies both $Y_2$ and $Y_3$: In this case, we have $Y_2 >_1 X_2$ and $Y_3 >_1 X_2$. We run the PR algorithm on the partition $(X_2, Y_2, Y_3)$ with agent 1's valuation function ($v_1(\cdot)$) and let agent 2 pick her favourite bundle from the final partition $X''$ returned by the PR algorithm. Then, we have a partition that satisfies the invariants and $\phi(X'')> \phi(X)$ as $\mathit{min}_1(X''_1, X''_2, X''_3) >_1 \mathit{min}_1(X_2,Y_2,Y_3) = X_2>_1 X_1 = \phi(X)$, where the first inequality follows from Lemma~\ref{PlautRoghgarden}.
    
    \item Agent 1 strongly envies one of $Y_2$ and $Y_3$: Let us assume without loss of generality that agent 1 strongly envies $Y_2$. Let $\overline{Y}_2$ be the minimal subset of $Y_2$ w.r.t. set inclusion that agent 1 values more than $X_2$.  Then, consider the partition $X'' = ( X''_1, X''_2, X''_3 )$  where $X''_1 = X_2$, $X''_2 = \overline{Y}_2$ and $X''_3 = Y_3 \cup (Y_2 \setminus \overline{Y}_2)$. First note that $X''_3$ is EFX-feasible  for agent 3 as $X'_3 = Y_3$ was EFX-feasible  in allocation $X'$ and now the bundle $X''_1$ remains the same, the bundle $X'_2$ has been compressed further in $X''$, and $X'_3 \subset X''_3$. Also note that $\phi(X'') = \mathit{min}(v_1(X''_1),v_1(X''_2)) =  \mathit{min}(v_1(X_2), v_1(\overline{Y}_2)) = v_1(X_2) > v_1(X_1) = \phi(X)$. If $X''_1$ and $X''_2$ are EFX-feasible  for agent 1, then partition $X''$ satisfies the invariants and $\phi(X'') > \phi(X)$ and we are done. So now consider the case when at least one of $X''_1$ and $X''_2$ is not EFX-feasible  for agent 1. Note that $X''_1 >_1 X''_2 \setminus h$ for all $h \in X''_2$ and $X''_2 >_1 X''_1$ by the fact that $X''_1 =X_2$ and by the definition of $X''_2 = \overline{Y}_2$. Thus, if one of $X''_1$ or $X''_2$ is not EFX-feasible  for agent 1, then we must have $X''_3 \setminus h' >_1 \mathit{min}_1(X''_1,X''_2)$ for some $h' \in X''_3$. In this case, we run the PR algorithm on the partition $(X''_1, X''_2, X''_3)$ with agent 1's valuation function $v_1( \cdot )$ and let agent 2 pick her favourite bundle from the final partition $Z$ returned by the PR algorithm. Then $Z$ satisfies the invariants and 
    \begin{align*}
        \phi(Z) &\geq \mathit{min}(v_1(Z_1), v_1(Z_2), v_1(Z_3))\\
                  &\geq \mathit{min}(v_1(X''_1), v_1(X''_2), v_1(X''_3))\\
                  &=v_1(X_2)\\
                  &>v_1(X_1)=\phi(X).
    \end{align*}
    So we are done.
\end{itemize}

This brings us to the main result of this section.

\begin{theorem}
\label{mainthm1}
Given an instance $I = \langle [3], M, \mathcal V \rangle$ such that $v_3(\cdot)$ is \Hana (no assumptions other than monotonicity on $v_1(\cdot)$ and $v_2(\cdot)$), there always exists an allocation $X = \langle X_1, X_2, X_3 \rangle$ such that $X$ is EFX. 
\end{theorem}

\section{Bounds on Rainbow Cycle Number}
\label{rainbowcycle}
In this section we improve the upper bounds on the rainbow cycle number introduced in~\cite{CGMMM21}, thereby implying the existence of approximate EFX allocations with $\mathcal{\tilde{O}}(n/ \varepsilon)^{\sfrac{1}{2}})$ charity.  \cite{CGMMM21} give an upper bound of $R(d) \in \mathcal{O}(d^4)$ and they show it results in the existence of a $(1-\varepsilon)$-EFX allocation with $\mathcal{O}((n/\varepsilon)^{\sfrac{4}{5}})$ charity. In the same paper,~\cite{CGMMM21} show a lower bound of $d$ on $R(d)$. In this section, we show improved bounds on $R(d)$. In particular, we first show in Section~\ref{boundonR} that $R(d) \in \mathcal{O}(d \log d)$ (making the upper bound almost tight), and thereby implying the existence of $(1-\varepsilon)$-EFX allocations with $\mathcal{\tilde{O}}((n / \varepsilon)^{\sfrac{1}{2}})$ charity. Secondly, in section \ref{permutation}, we show an upper bound of $2d-2$ assuming that every vertex $v \in V_i$ has exactly one incoming edge from any other part $V_j \neq V_i$ and exactly one outgoing edge to some vertex in $V_j$. We call this number $R_p(d)$. We remark that the lower bound of $d$ in~\cite{CGMMM21} also holds for $R_p(d)$. The upper bound of $2d-2$ immediately improves the upper-bound on the zero-sum extremal problem studied in~\cite{AK21, MS21}.


\subsection{Almost Tight Upper Bound on $R(d)$}
\label{boundonR}
Recall that $R(d)$ is the largest
$k$ such that there exists a 
$k$-partite digraph $G$ with $k$ classes of vertices $V_i$ so that each part
$V_i$ has at most $d$ vertices, for all distinct $i,j$ each 
vertex in $V_i$ has an incoming  edge
from some vertex in $V_j$ and vice versa, and 
there exists no (directed) rainbow cycle, namely, a cycle
in $G$ 
that contains at most one vertex of each $V_i$. 
In this section, we prove the following improved bound which is
tight up to the logarithmic factor. 
\begin{theorem}
\label{mainthmsec2}
\label{p11}
If 
\begin{equation}
\label{e11}
k (1-1/d)^{k-1} <1
\end{equation}
then $R(d)<k.$
Therefore
$R(d)\leq (1+o(1))d \log d$
\end{theorem}
\begin{proof}
Suppose $k (1-1/d)^{k-1} <1$.
Let $S$ be a random set of $k$ vertices of $G$ obtained by picking
a single  vertex $v_i$ in each $V_i$, randomly and uniformly among all
vertices of $V_i$, where all choices  are independent.  For each
vertex $v$, let $E_v$ be the event that $S$ contains $v$ and 
contains no other vertex
$u$ so that $uv$ is a directed edge. We claim that if $v\in V_i$ then the
probability of $E_v$ is at most
$$
\frac{1}{|V_i|} (1-1/d)^{k-1}.
$$
Indeed, the probability that $v \in S$ is $1/|V_i|$. Conditioning on that,
since for every $j \neq i$ there is some 
$u_j \in  V_j$ so that $u_jv$ is a directed edge, and the probability
that $u_j$ is in $S$ is $1/|V_j| \geq 1/d$, the probability that $v$ has
non in-neighbor  in $V_j$ is at most $1-1/d$.
As the choices are independent, the claim follows.
By the union bound, the probability, that there is a vertex $v$ so that 
the event $E_v$ occurs is at most
$$
\sum_{i=1}^k |V_i| \frac{1}{|V_i|} (1-1/d)^{k-1} =k(1-1/d)^{k-1}<1.
$$
Therefore, with positive probability every vertex in the induced
subgraph of $G$ on $S$ has an in-neighbor. Hence there is such an $S$
and in this induced subgraph there is a cycle which contains at most
one vertex from each $V_i$.
Thus $R(d)<k$, completing the proof.  
\end{proof}
\vspace{0.2cm}

Theorems~\ref{theorem-reduction} and Theorem~\ref{mainthmsec2} then imply Theorem~\ref{eps-efx}.

\noindent
\paragraph{Remark.} The proof above can be derandomized using the
method of conditional expectations
(cf., e.g.,~\cite{AlonS92}, chapter 16), giving the following.
\begin{proposition}
\label{p12}
Let $G$ be a $k$-partite digraph with classes of vertices 
$V_i$, each having at most
$d$ vertices. Suppose that for all distinct $i,j$ each 
vertex in $V_i$ has an incoming  edge
from some vertex in $V_j$ and vice versa, 
and suppose that (\ref{e11}) holds. Then a 
rainbow cycle in $G$ can be found by a deterministic 
polynomial time algorithm.
\end{proposition}
\begin{proof}
We apply the method of conditional expectations to produce a set 
$S=\{s_1, s_2, \ldots ,s_k\}$ of vertices of $G$, where
$s_i \in V_i$, so that every
indegree in the induced subgraph of $G$ on $S$ is positive. This is done
by choosing the vertices $s_i$ one by one, in order, maintaining
a potential function $\phi$ whose value is the conditional 
expectation of the number of events $E_v$ that hold, given the
choices of the vertices $s_i$ made so far.

At the beginning, there are no choices made, and the value of
$\phi$ is the sum
$$
\sum_{i=1}^k |V_i| \frac{1}{|V_i|} (1-1/d)^{k-1} =k(1-1/d)^{k-1}<1.
$$
Assuming $s_1,s_2, \ldots ,s_{i-1}$ have already been chosen
and the above conditional expectation is still smaller than
$1$, choose $s_i \in V_i$ to be the vertex that minimizes the
updated value of the conditional expectation. As the expectation is the
average over all possible choices of $s_i$, this minimum stays below $1$.
The computation of the required conditional expectations, for each of
the possible $|V_i|\leq d$ choices of $s_i \in V_i$, can clearly be
done efficiently. At the end of the process the value of the potential
function is exactly the number of events $E_v$ that hold, and since this
number is below $1$, none of them holds. This supplies the required
set $S$. Starting in any vertex of $S$ and moving repeatedly to an
in-neighbor of it in $S$ until we reach a vertex we have already visited
supplies the desired rainbow cycle.
\end{proof}

\label{general}
\subsection{A linear upper bound on $R_p(d)$}
\label{perm-rainbowcyclenumer}
In this section we assume graph $G$ satisfies all the properties in Definition \ref{rainbow-def} and also for all different parts $V_i$ and $V_j$, each vertex in $V_i$ has exactly one outgoing edge to a vertex in $V_j$. We call these graphs permutation graphs since the set of edges from any part to any other part defines a permutation.

\begin{definition}\label{rainbow-perm-def}
    Given an integer $d > 0$, the permutation rainbow cycle number $R_p(d)$ is the largest $k$ such that there exists a $k$-partite graph $G =(V_1 \cup V_2 \cup \dots \cup V_k, E)$ such that 
    \begin{itemize}
        \item each part has exactly $d$ vertices, i.e., $\lvert V_i \rvert = d$, and 
        \item every vertex has exactly one incoming edge from every part other than the one containing it.
        \item every vertex has exactly one outgoing edge to every part other than the one containing it.
        \item  there exists no cycle $C$ in $G$ that visits each part at most once.
    \end{itemize}
\end{definition}

\begin{theorem}\label{perm-theorem}
    For all integers $d>0$, $R_p(d)<2d-1$.
\end{theorem}
 In the rest of this section we prove Theorem \ref{perm-theorem}. The proof is by induction. 
 
\paragraph{Basis:} For the base case, consider $d=1$. If there are $2$ parts or more, the vertex in $V_1$ has an outgoing edge to the vertex in $V_2$ and vice versa. Therefore, there exists a rainbow cycle $C$ in $G$ which is a contradiction. Thus, $R_p(1)<2$. 
 
\paragraph{Induction step:} We assume 
\begin{align}
    \text{for all $d'<d$, \quad $R_p(d')<2d'-1$,}\label{induc-assump}
\end{align}
and prove $R_p(d)<2d-1$. First we define $i$-restricted paths which are the paths that use each part at most once and except for the last vertex, all vertices are in the first $i$ parts.

\begin{definition}
    We call path $P = v_1 \rightarrow v_2 \rightarrow \dots \rightarrow v_t$ an $i$-restricted path if 
    \begin{itemize}
        \item $v_1, \ldots, v_{t-1} \in V_1 \cup V_2 \cup \dots \cup V_{i}$, and
        \item $P$ visits each part at most once.
    \end{itemize}
\end{definition}
Note that for all $j>i$, every $i$-restricted path is also a $j$-restricted path. 
Now we prove the following claim.
\begin{claim}\label{colorful-claim}
    If $k \geq 2d-1$, for every vertex $v$, there is a way of reindexing the parts such that $v \in V_1$ and for all $i \in [d]$, there are $i$ nodes in $V_{2i-1}$ which are reachable from $v$ via $(2i-2)$-restricted paths.
\end{claim}
\begin{proof}
    The proof of the claim is also by induction. For the base case let $i=1$. If $v \in U$, set $V_1 = U$ and the claim follows.
    For the induction step, we assume $V_1, V_2, \dots, V_{2i-1}$ are already defined and there is a $(2i-2)$-restricted path from $v$ to $v_1, v_2, \dots, v_i \in V_{2i-1}$. Consider any part $U \notin \{V_1, V_2, \dots, V_{2i-1}\}$. For all $j \in [i]$, let $v_j \rightarrow u_j$ be the outgoing edge from $v_j$ to $U$. Since each node in $V_{2i-1}$ has exactly one outgoing edge to $U$ and each node in $U$ has exactly one incoming edge from $V$, $u_1, u_2, \ldots, u_i$ are distinct. Therefore, at least $i$ nodes in $U$ are reachable from $v$ via $(2i-1)$-restricted paths. Let $U' \subseteq U$ be the vertices that are reachable from $v$ via $(2i-1)$-restricted paths and let $\overline{U} = U \setminus U'$. If $\lvert U'\rvert \geq i+1$, we set $V_{2i} = W$ for some $W \notin \{V_1, V_2, \dots, V_{2i-1}, U\}$ and set $V_{2i+1} = U$ and the claim follows. Otherwise, for all $U \notin \{V_1, V_2, \dots, V_{2i-1}\}$, we have $\lvert U'\rvert = i$ and $\lvert \overline{U}\rvert = d-i$. If there exist $U,W \notin \{V_1, V_2, \dots, V_{2i-1}\}$ such that $w \in W'$ has an outgoing edge to $u \in \overline{U}$, then we set $V_{2i} = W$ and $V_{2i+1} = U$. Note that all nodes in $U'$ are reachable from $v$ using $(2i-1)$-restricted paths and $u$ is reachable via a $(2i)$-restricted path. Therefore, in total $i+1$ vertices in $U = V_{2i+1}$ are reachable from $v$ via $(2i)$-restricted paths. See Figure \ref{colorful-fig-1} for an illustration.
    \begin{figure}
        \centering
        \begin{tikzpicture}
[scale=0.85,
 agent/.style={circle, draw=green!60, fill=green!5, very thick},
 good/.style={circle, draw=red!60, fill=red!5, very thick, minimum size=1pt},
]

\draw[black, very thick] (-0.5,-0.5-1) rectangle (0.5,3.5);

\node[agent]      (v) at (0,2.75)      {$\scriptstyle{v}$};

\filldraw[color=black!60, fill=black!5, very thick](0, 2) circle (0.02);
\filldraw[color=black!60, fill=black!5, very thick](0, 1.5) circle (0.02);
\filldraw[color=black!60, fill=black!5, very thick](0, 1) circle (0.02);

\node at (0, -0.5-1.30) {$\scriptstyle{V_1}$};

\draw[black, very thick] (-0.5+4,-0.5-1) rectangle (0.5+4,3.5);

\node[agent]      (v1) at (4,2.75)      {$\scriptstyle{v_1}$};
\node[agent]      (vi) at (4,2.75-1.75)      {$\scriptstyle{v_{i}}$};

\filldraw[color=black!60, fill=black!5, very thick](4, 2) circle (0.02);
\filldraw[color=black!60, fill=black!5, very thick](4, 1.6) circle (0.02);

\node at (4, -0.5-1.30) {$\scriptstyle{V_{2i-1}}$};

\draw[black, very thick] (-0.6+5.5,-0.5-1) rectangle (0.6+5.5,3.5);

\draw[black] (-0.5+5.6,0.6) rectangle (0.5+5.4,3.1);
\node at (5.5, 3.3) {$\scriptstyle{W'}$};
\node[agent]      (w) at (5.5,2.75-1.75)      {$\scriptstyle{w}$};

\filldraw[color=black!60, fill=black!5, very thick](5.5, 2) circle (0.02);
\filldraw[color=black!60, fill=black!5, very thick](5.5, 1.6) circle (0.02);

\draw[black] (-0.5+5.6,-1) rectangle (0.5+5.4,0.5);
\node at (5.5, -1.22) {$\scriptstyle{\overline{W}}$};

\filldraw[color=black!60, fill=black!5, very thick](5.5, 0) circle (0.02);
\filldraw[color=black!60, fill=black!5, very thick](5.5, -0.4) circle (0.02);

\node at (5.5, -0.5-1.30) {$\scriptstyle{V_{2i}=W}$};

\draw[black, very thick] (-0.6+7,-0.5-1) rectangle (0.6+7,3.5);

\draw[black] (1+5.6,0.6) rectangle (2+5.4,3.1);
\node at (5.5+1.5, 3.3) {$\scriptstyle{U'}$};

\draw[black] (1+5.6,-1) rectangle (2+5.4,0.5);
\node at (7, -1.22) {$\scriptstyle{\overline{U}}$};

\node[agent]      (u) at (7,2.75-2.7)      {$\scriptstyle{u}$};

\filldraw[color=black!60, fill=black!5, very thick](7, 2) circle (0.02);
\filldraw[color=black!60, fill=black!5, very thick](7, 1.6) circle (0.02);

\node at (7, -0.5-1.30) {$\scriptstyle{V_{2i+1}=U}$};

\draw[black,->] (w)--(u);

\filldraw[color=black!60, fill=black, thick](1.5, 1.5) circle (0.01);
\filldraw[color=black!60, fill=black, thick](2, 1.5) circle (0.01);
\filldraw[color=black!60, fill=black, thick](2.5, 1.5) circle (0.01);
\end{tikzpicture}
        \caption{$W'$ has an outgoing edge to $\overline{U}$}
        \label{colorful-fig-1}
    \end{figure}
    
    Let $V(G) = V_1 \cup V_2 \cup \dots \cup V_{2i-1} \cup U_1 \cup U_2 \cup \dots \cup U_{k-2i+1}$. The only remaining case is that for all $j \in [k-2i+1]$, $\lvert \overline{U}_j\rvert = d-i$ and for all $j, \ell \in [k-2i+1]$, there is no edge from $U'_j$ to $\overline{U}_{\ell}$. This means that all the $d-i$ incoming edges of $\overline{U}_{\ell}$ come from $\overline{U}_{j}$. Hence all the $d-i$ outgoing edges of $\overline{U}_{j}$ go to $\overline{U}_{\ell}$. Therefore, the induced subgraph on $\overline{U}_1 \cup \overline{U}_2 \cup \dots \cup \overline{U}_{k-2i+1}$, forms a permutation graph. See Figure \ref{colorful-fig-2}. By Inequality \eqref{induc-assump}, we know $R_p(d-i) < 2d-2i-1$ and hence, $k-2i+1 < 2d-2i-1$. This is a contradiction with the assumption of the claim which requires $k \geq 2d-1$. Therefore, this case cannot occur.
    \begin{figure}
        \centering
        \begin{tikzpicture}
[scale=0.85,
 agent/.style={circle, draw=green!60, fill=green!5, very thick},
 good/.style={circle, draw=red!60, fill=red!5, very thick, minimum size=1pt},
]

\draw[black, very thick] (-0.6,-0.5-1) rectangle (0.6,3.5);

\draw[red] (-0.5+0.1,0.6) rectangle (0.5-0.1,3.1);
\node at (0, 3.3) {$\scriptstyle{U'_1}$};

\filldraw[color=black!60, fill=black!5, very thick](0, 2) circle (0.02);
\filldraw[color=black!60, fill=black!5, very thick](0, 1.6) circle (0.02);

\draw[blue] (-0.5+0.1,-1) rectangle (0.5-0.1,0.5);
\node at (0, -1.22) {$\scriptstyle{\overline{U}_1}$};

\filldraw[color=black!60, fill=black!5, very thick](0, 0) circle (0.02);
\filldraw[color=black!60, fill=black!5, very thick](0, -0.4) circle (0.02);

\draw[red, thick, ->]  (0,2)--(1.5, 1.6);
\draw[red, thick, ->]  (0,1.6)--(1.5, 2);

\draw[blue, thick, ->]  (0,0)--(1.5, -0.4);
\draw[blue, thick, ->]  (0,-0.4)--(1.5, 0);

\node at (0, -0.5-1.30) {$\scriptstyle{U_1}$};

\draw[black, very thick] (-0.6+1.5,-0.5-1) rectangle (0.6+1.5,3.5);

\draw[red] (1+0.1,0.6) rectangle (2-0.1,3.1);
\node at (0+1.5, 3.3) {$\scriptstyle{U'_2}$};


\filldraw[color=black!60, fill=black!5, very thick](1.5, 2) circle (0.02);
\filldraw[color=black!60, fill=black!5, very thick](1.5, 1.6) circle (0.02);

\draw[blue] (1+0.1,-1) rectangle (2-0.1,0.5);
\node at (1.5, -1.22) {$\scriptstyle{\overline{U}_2}$};

\filldraw[color=black!60, fill=black!5, very thick](1.5, 0) circle (0.02);
\filldraw[color=black!60, fill=black!5, very thick](1.5, -0.4) circle (0.02);

\node at (1.5, -0.5-1.30) {$\scriptstyle{U_2}$};

\draw[black, very thick] (-0.6+5.5,-0.5-1) rectangle (0.6+5.5,3.5);

\draw[red] (5+0.1,0.6) rectangle (6-0.1,3.1);
\node at (4+1.5, 3.3) {$\scriptstyle{U'_{k'}}$};


\filldraw[color=black!60, fill=black!5, very thick](5.5, 2) circle (0.02);
\filldraw[color=black!60, fill=black!5, very thick](5.5, 1.6) circle (0.02);

\draw[blue] (5+0.1,-1) rectangle (6-0.1,0.5);
\node at (5.5, -1.22) {$\scriptstyle{\overline{U}_{k'}}$};

\filldraw[color=black!60, fill=black!5, very thick](5.5, 0) circle (0.02);
\filldraw[color=black!60, fill=black!5, very thick](5.5, -0.4) circle (0.02);

\node at (5.5, -0.5-1.30) {$\scriptstyle{U_{k'}}$};

\draw[red, thick, ->]  (4.5,2)--(5.5, 2);
\draw[red, thick, ->]  (4.5,1.6)--(5.5, 1.6);

\draw[blue, thick, ->]  (4.5,0)--(5.5, 0);
\draw[blue, thick, ->]  (4.5,-0.4)--(5.5, -0.4);

\filldraw[color=black!60, fill=black, thick](2.75, 1.5) circle (0.01);
\filldraw[color=black!60, fill=black, thick](3.25, 1.5) circle (0.01);
\filldraw[color=black!60, fill=black, thick](3.75, 1.5) circle (0.01);
\end{tikzpicture}
        \caption{$k' \geq k-2i-1$ and for all $j, \ell \in [k']$, there exists no edge between $U'_j$ and $\overline{U}_{\ell}$.}
        \label{colorful-fig-2}
    \end{figure}
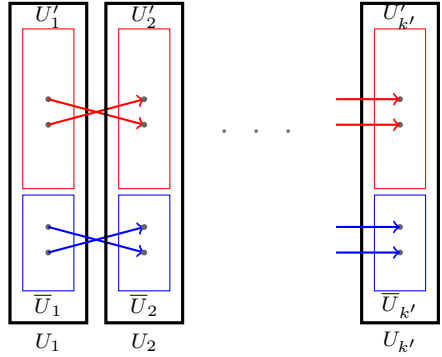
\end{proof}
Back to the assumption step, we want to prove $R_p(d)<2d-1$. Towards a contradiction, assume $R_p(d) \geq 2d-1$ and consider a graph $G$ with $\lvert R_p(d)\rvert$ parts satisfying properties of Definition \ref{rainbow-perm-def}. Now pick an arbitrary vertex $v$. By setting $i=d$ in Claim \ref{colorful-claim}, there exists a reindexing of the parts such that all $d$ nodes in part $V_{2d-1}$ are reachable from $v$ using $(2d-2)$-restricted paths. Let $u \in V_{2d-1}$ be the vertex with an outgoing edge to $v$. Then a $(2d-2)$-restricted path from $v$ to $u$ followed by the edge $u \rightarrow v$ forms a rainbow cycle. Hence, $R_p(d)<2d-1$.\label{permutation}

\bibliographystyle{alpha}
\bibliography{sample}

\newcommand{\etalchar}[1]{$^{#1}$}
\begin{thebibliography}{AGMW15}

\bibitem[AGMW15]{AzizGMW15}
Haris Aziz, Serge Gaspers, Simon Mackenzie, and Toby Walsh.
\newblock Fair assignment of indivisible objects under ordinal preferences.
\newblock {\em Artif. Intell.}, 227:71--92, 2015.

\bibitem[AK21]{AK21}
Noga Alon and Michael Krivelevich.
\newblock Divisible subdivisions.
\newblock {\em J. Graph Theory}, 98(4):623--629, 2021.

\bibitem[AMN20]{AmanatidisMN20}
Georgios Amanatidis, Evangelos Markakis, and Apostolos Ntokos.
\newblock Multiple birds with one stone: Beating 1/2 for {EFX} and {GMMS} via
  envy cycle elimination.
\newblock {\em Theor. Comput. Sci.}, 841:94--109, 2020.

\bibitem[AMNS17]{AMNS17}
Georgios Amanatidis, Evangelos Markakis, Afshin Nikzad, and Amin Saberi.
\newblock Approximation algorithms for computing maximim share allocations.
\newblock {\em ACM Transactions on Algorithms}, 13(4):52:1--52:28, 2017.

\bibitem[AS92]{AlonS92}
Noga Alon and Joel Spencer.
\newblock {\em The Probabilistic Method}.
\newblock John Wiley, 1992.

\bibitem[BBC10]{BudishC10}
Eric B.~Budish and Estelle Cantillon.
\newblock The multi-unit assignment problem: Theory and evidence from course
  allocation at {H}arvard.
\newblock {\em American Economic Review}, 102, 2010.

\bibitem[BBK22]{BBK22}
Benjamin~Aram Berendsohn, Simona Boyadzhiyska, and László Kozma.
\newblock Fixed-point cycles and {EFX} allocations.
\newblock {\em CoRR}, 2201.08753, 2022.

\bibitem[BCFF21]{BergerCFF}
Ben Berger, Avi Cohen, Michal Feldman, and Amos Fiat.
\newblock ({A}lmost full) {EFX} exists for four agents (and beyond).
\newblock {\em CoRR}, abs/2102.10654, 2021.

\bibitem[BGGS21]{BaklanovGGS21}
Artem Baklanov, Pranav Garimidi, Vasilis Gkatzelis, and Daniel Schoepflin.
\newblock Achieving proportionality up to the maximin item with indivisible
  goods.
\newblock In {\em Thirty-Fifth {AAAI} Conference on Artificial Intelligence,
  {AAAI}}, pages 5143--5150, 2021.

\bibitem[BK17]{BK17}
Siddharth Barman and Sanath~Kumar Krishnamurthy.
\newblock Approximation algorithms for maximin fair division.
\newblock In {\em Proceedings of the 18th ACM Conference on Economics and
  Computation (EC)}, pages 647--664, 2017.

\bibitem[BKK17]{BramsKK17}
Steven~J. Brams, D.~Marc Kilgour, and Christian Klamler.
\newblock Maximin envy-free division of indivisible items.
\newblock {\em Group Decision and Negotiation}, 26(1):115--131, 2017.

\bibitem[BKV18]{BKV18}
Siddharth Barman, Sanath~Kumar Krishnamurthy, and Rohit Vaish.
\newblock Finding fair and efficient allocations.
\newblock In {\em Proceedings of the 19th ACM Conference on Economics and
  Computation (EC)}, pages 557--574, 2018.

\bibitem[BL16]{BL16}
Sylvain Bouveret and Michel Lema\^itre.
\newblock Characterizing conflicts in fair division of indivisible goods using
  a scale of criteria.
\newblock In {\em Autonomous Agents and Multi-Agent Systems (AAMAS) 30, 2},
  pages 259--290, 2016.

\bibitem[BT96]{BramsT96}
Steven~J. Brams and Alan~D. Taylor.
\newblock {\em Fair division - from cake-cutting to dispute resolution}.
\newblock Cambridge University Press, 1996.

\bibitem[Bud11]{budish2011combinatorial}
Eric Budish.
\newblock The combinatorial assignment problem: Approximate competitive
  equilibrium from equal incomes.
\newblock {\em Journal of Political Economy}, 119(6):1061--1103, 2011.

\bibitem[CGH19]{CaragiannisGravin19}
Ioannis Caragiannis, Nick Gravin, and Xin Huang.
\newblock Envy-freeness up to any item with high {N}ash welfare: The virtue of
  donating items.
\newblock In {\em Proceedings of the 20th ACM Conference on Economics and
  Computation (EC)}, pages 527--545, 2019.

\bibitem[CGM20]{ChaudhuryGM20}
Bhaskar~Ray Chaudhury, Jugal Garg, and Kurt Mehlhorn.
\newblock {EFX} exists for three agents.
\newblock In {\em Proc.\ 21st Conf.\ Economics and Computation (EC)}, pages
  1--19. {ACM}, 2020.

\bibitem[CGM{\etalchar{+}}21]{CGMMM21}
Bhaskar~Ray Chaudhury, Jugal Garg, Kurt Mehlhorn, Ruta Mehta, and Pranabendu
  Misra.
\newblock Improving {EFX} guarantees through rainbow cycle number.
\newblock In {\em Proceedings of the 22nd ACM Conference on Economics and
  Computation (EC)}, pages 310--311. {ACM}, 2021.

\bibitem[CKM{\etalchar{+}}16]{CaragiannisKMP016}
Ioannis Caragiannis, David Kurokawa, Herv{\'{e}} Moulin, Ariel~D. Procaccia,
  Nisarg Shah, and Junxing Wang.
\newblock The unreasonable fairness of maximum {Nash} welfare.
\newblock In {\em Proceedings of the 17th ACM Conference on Economics and
  Computation (EC)}, pages 305--322, 2016.

\bibitem[CKMS21]{CKMS21}
Bhaskar~Ray Chaudhury, Telikepalli Kavitha, Kurt Mehlhorn, and Alkmini
  Sgouritsa.
\newblock A little charity guarantees almost envy-freeness.
\newblock {\em {SIAM} J. Comput.}, 50(4):1336--1358, 2021.

\bibitem[EPT05]{EtkinPT05}
R.~Etkin, A.~Parekh, and D.~Tse.
\newblock Spectrum sharing for unlicensed bands.
\newblock In {\em In Proceedings of the first IEEE Symposium on New Frontiers
  in Dynamic Spectrum Access Networks}, 2005.

\bibitem[fai]{fairoutcome}
\url{www.fairoutcomes.com}.

\bibitem[GHS{\etalchar{+}}18]{GhodsiHSSY18}
Mohammad Ghodsi, Mohammad~Taghi Hajiaghayi, Masoud Seddighin, Saeed Seddighin,
  and Hadi Yami.
\newblock Fair allocation of indivisible goods: Improvements and
  generalizations.
\newblock In {\em Proceedings of the 19th {ACM} Conference on Economics and
  Computation (EC)}, pages 539--556, 2018.

\bibitem[GMT19]{JGargMT19}
Jugal Garg, Peter McGlaughlin, and Setareh Taki.
\newblock Approximating maximin share allocations.
\newblock In {\em Proceedings of the 2nd Symposium on Simplicity in Algorithms
  (SOSA)}, volume~69, pages 20:1--20:11, 2019.

\bibitem[GT20]{GargT19}
Jugal Garg and Setareh Taki.
\newblock An improved approximation algorithm for maximin shares.
\newblock In {\em {EC}}, pages 379--380. {ACM}, 2020.

\bibitem[HPPS20]{HalperPPS20}
Daniel Halpern, Ariel~D. Procaccia, Alexandros Psomas, and Nisarg Shah.
\newblock Fair division with binary valuations: One rule to rule them all.
\newblock In {\em {WINE}}, volume 12495 of {\em Lecture Notes in Computer
  Science}, pages 370--383. Springer, 2020.

\bibitem[KPW18]{KPW18}
David Kurokawa, Ariel~D. Procaccia, and Junxing Wang.
\newblock Fair enough: Guaranteeing approximate maximin shares.
\newblock {\em Journal of ACM}, 65(2):8:1--27, 2018.

\bibitem[LMMS04]{LiptonMMS04}
Richard~J. Lipton, Evangelos Markakis, Elchanan Mossel, and Amin Saberi.
\newblock On approximately fair allocations of indivisible goods.
\newblock In {\em Proc.\ 5th Conf.\ Economics and Computation (EC)}, pages
  125--131, 2004.

\bibitem[Mou19]{Moulin19}
Hervé Moulin.
\newblock Fair division in the internet age.
\newblock {\em Annual Review of Economics}, 11(1):407--441, 2019.

\bibitem[MS21]{MS21}
Tam{\'{a}}s M{\'{e}}sz{\'{a}}ros and Raphael Steiner.
\newblock Zero sum cycles in complete digraphs.
\newblock {\em Eur. J. Comb.}, 98:103399, 2021.

\bibitem[PR20]{TimPlaut18}
Benjamin Plaut and Tim Roughgarden.
\newblock Almost envy-freeness with general valuations.
\newblock {\em {SIAM} J. Discret. Math.}, 34(2):1039--1068, 2020.

\bibitem[Pro20]{ProcacciaCACM20}
Ariel~D. Procaccia.
\newblock Technical perspective: An answer to fair division's most enigmatic
  question.
\newblock {\em Commun. ACM}, 63(4):118, March 2020.

\bibitem[PZ90]{PrattZ90}
John~Winsor Pratt and Richard~Jay Zeckhauser.
\newblock The fair and efficient division of the {W}insor family silver.
\newblock {\em Management Science}, 36(11):1293--1301, 1990.

\bibitem[spl]{spliddit}
\url{www.spliddit.org}.

\bibitem[Ste48]{Steinhaus48}
Hugo Steinhaus.
\newblock The problem of fair division.
\newblock {\em Econometrica}, 16(1):101--104, 1948.

\bibitem[Vos02]{Vossen02}
T.~W.M. Vossen.
\newblock {\em Fair allocation concepts in air traffic management}.
\newblock PhD thesis, University of Maryland, College Park, 2002.

\end{thebibliography}

\appendix
\section{Appendix}

\paragraph{Non-Degenerate Instances~\cite{ChaudhuryGM20}.} We call an instance $I = \langle [n], M, \mathcal{V} \rangle$ non-degenerate if and only if no agent values two different sets equally, i.e., $\forall i \in [n]$ we have $v_i(S) \neq v_i(T)$ for all $S \neq T$. We extend the technique in~\cite{ChaudhuryGM20} and  show that it suffices to deal with non-degenerate instances when there are $n$ agents with general valuation functions, i.e., if there exists an EFX allocation in all non-degenerate instances, then there exists an EFX allocation in all instances. 

Let $M = \left\{g_1,g_2,\dots,g_m\right\}$. We perturb any instance $I$ to $I(\varepsilon) = \langle [n],M ,\mathcal{V}(\varepsilon) \rangle$, where for every $v_i \in \mathcal{V}$ we define $v'_i \in \mathcal{V}(\varepsilon)$, as

$$ v'_i(S) = v_i(S) + \varepsilon \cdot \sum_{g_j \in S} 2^{j} \quad \quad \forall S \subseteq M$$ 

\begin{lemma}
	\label{non-degeneracy-technical} \label{non-degeneracy-main}
	Let $\delta = 
\min_{i \in [n]} \min_{S,T \colon v_i(S) \neq v_i(T)} \abs{ v_i(S) - v_i(T)}$ 
	and let $\varepsilon > 0$ be such that $\varepsilon \cdot 2^{m+1}  < \delta$. Then
	\begin{enumerate}
		\item For any agent $i$ and $S,T \subseteq M$ such that $v_i(S) > v_i(T)$, we have $v'_i(S) > v'_i(T)$.
		\item $I(\varepsilon)$ is a non-degenerate instance. Furthermore, if $X = \langle X_1,X_2,X_3 \rangle$ is an EFX allocation for $I(\varepsilon)$ then $X$ is also an EFX allocation for $I$.
	\end{enumerate}
\end{lemma}
\begin{proof}
	For the first statement of the lemma, observe that 
	\begin{align*}
	v'_i(S) - v'_i(T)  &=  v_i(S) - v_i(T)  + \varepsilon(\sum_{g_j \in S \setminus T}2^j - \sum_{g_j \in T \setminus S}2^j) \\
	&\geq \delta -  \varepsilon \sum_{g_j \in T \setminus S}2^j\\
	&\geq \delta -  \varepsilon \cdot (2^{m+1}-1)\\
	&>0 \enspace .  
	\end{align*}
	
	For the second statement of the lemma, consider any two sets $S,T \subseteq M$ such that $S \neq T$. Now, for any $i \in [n]$, if $v_i(S) \neq v_i(T)$, we have $v'_i(S) \neq v'_i(T)$ by the first statement of the lemma. If $v_i(S) = v_i(T)$, we have $v'_i(S) - v'_i(T) = \varepsilon(\sum_{g_j \in S \setminus T}2^j - \sum_{g_j \in T \setminus S}2^j) \neq 0$ (as $S \neq T$). Therefore, $I(\varepsilon)$ is non-degenerate.
	
	For the final claim, let us assume that $X$ is an EFX allocation in $I(\varepsilon)$ and not an EFX allocation in $I$. Then there exist $i,j$, and $g \in X_j$ such that $v_i(X_j \setminus g) > v_i(X_i)$. In that case, we have $v'_i(X_j \setminus g) > v'_i(X_i)$ by the first statement of the lemma, implying that $X$ is not an EFX allocation in $I(\varepsilon)$ as well, which is a contradiction. 
\end{proof}

\end{document}